\documentclass{sig-alternate}

\usepackage{stmaryrd}
\usepackage{cite}
\usepackage{paralist}
\usepackage{amsmath}
\usepackage{amsthm}
\usepackage{amsfonts}
\usepackage{graphicx}
\usepackage{epstopdf}
\usepackage{subcaption}
\usepackage{subcaption}
\usepackage{listings}
\usepackage{paralist}
\usepackage{tikz}
\usepackage[ruled,commentsnumbered,linesnumbered]{algorithm2e}
\usepackage{lipsum}
\usepackage{hyperref}
\usepackage{url}
\usepackage{multicol}
\usepackage{blindtext}
\usepackage{array}
\usepackage{floatrow}
\usepackage{thmtools}
\usepackage{caption}
\usepackage{lipsum}

\hypersetup{
    bookmarks=true,         
    unicode=false,          
    pdftoolbar=true,        
    pdfmenubar=true,        
    pdffitwindow=false,     
    pdfstartview={FitH},    
    pdftitle={My title},    
    pdfauthor={Author},     
    pdfsubject={Subject},   
    pdfcreator={Creator},   
    pdfproducer={Producer}, 
    pdfkeywords={keyword1} {key2} {key3}, 
    pdfnewwindow=true,      
    colorlinks=true,       
    linkcolor=blue,          
    citecolor=blue,        
    filecolor=magenta,      
    urlcolor=blue           
}

\input{prelude1}

\newcommand{\sembrack}[1]{[\![#1]\!]}

\declaretheoremstyle[headfont=\bf]{normalhead}

\begin{document}

\title{Proving Abstractions of Dynamical Systems through Numerical Simulations}


 \numberofauthors{2}
\author{Sayan Mitra
\end{tabular} 
\\
\begin{tabular}{c}
\affaddr{mitras@illinois.edu}\\
\affaddr{Coordinate Science Laboratory}\\
\affaddr{University of Illinois at Urbana Champaign}\\
\affaddr{Urbana, IL 61801} \\
}

\maketitle

\begin{abstract}
A key question that arises in  rigorous analysis of cyberphysical systems under attack involves establishing whether or not  the attacked  system deviates significantly from the ideal allowed behavior. This is the problem of deciding whether or not the 
ideal system is an abstraction of the attacked system.
A quantitative variation of this question can capture how much the attacked system deviates from the ideal.
Thus, algorithms for deciding abstraction relations can help measure the effect of attacks
on cyberphysical systems and  to develop attack detection strategies. 
In this paper, we present a decision procedure for proving that one nonlinear dynamical system is a quantitative abstraction of another. Directly computing the reach sets of these nonlinear systems are undecidable in general and reach set over-approximations do not give a direct way for proving abstraction. Our procedure uses (possibly inaccurate) numerical simulations and a model annotation to compute tight approximations of the observable behaviors of the system and then uses these approximations to decide on abstraction. 
We show that the procedure is sound and that it is guaranteed to terminate under reasonable robustness assumptions. 
\end{abstract}

\section{Introduction}
\label{sec:intro}

Cyberphysical systems can take the form of anti-lock braking systems (ABS) in cars, process control systems in factories implemented over SCADA, all the way to city and nation-scale networked control systems for traffic, water, and power. 
Security breaches in cyberphysical systems can be disastrous. 
Aside from the obvious social motivation, an inquiry into the security of cyberphysical systems is also 
propelled by  new scientific questions about architechting and understanding computing systems that control the physical world.
As these computing systems are embedded in the physical world  (a) they require preservation of dynamical properties 
that cannot be characterized purely in terms of software state, and 
(b)  they  can be breached in ways that go beyond vulnerabilities that are exploited in stand alone computing systems.
While the dynamical operation remains vulnerable to full-fledged attacks  on its
computing and the communication components---for instance, a denial-of-service-attack on the computers controlling the power grid can take it down---it is
also vulnerable to more elusive {\em dynamics-aware attacks\/} that subtly change local behaviors in ways that lead to instability, unsafe behavior, and a loss of availability of the system.
In this paper, we present new results that contribute towards our  longer term goal of developing a framework for analyzing security properties of cyberphysical under different classes of attacks.



\paragraph*{Role of Models and Abstractions}
Design of control software begins with a mathematical model for the underlying physical process~\cite{luenberger,khalil-book-3ed}. 
which is usually described in the language of ordinary differential equations (ODEs). 
Any meaningful notion of attack, safety, resilience, availability, and performance, therefore, 
has to be expressed in this language. 
Indeed, our analysis framework is designed for analyzing models of cyberphysical systems 
that combine these ODEs with automaton models that are  used for representing computations.

A model $B$ is said to be an {\em abstraction\/} of another model $A$
if every observable behavior of $A$ is also an observable behavior of $B$~\cite{clarke1994model,TIOAmon}.
The abstract model $B$ could capture desired properties.
Here are two example properties:
``Alarm must go off $6$ seconds before car gets within $4 m$ of obstacle even if the position sensors are jammed'' (safety),
``Voltage remains within the range $B$ and eventually converges to the smaller range $B'$'' (invariance and progress). 
. 
Establishing that it is an abstraction of $A$ implies that all behaviors of $A$ satisfy these properties. 
This then enables us to substitute $A$ with $B$ when we are  analyzing a lager system (containing $A$), 
in which only these properties of $A$ are relevant.
The abstract model $B$ has more behaviors and is typically simpler to analyze than the concrete model $A$.
In some extreme cases, the abstract model lends itself to  completely algorithmic analysis even though 
the concrete model does not (see, for example~\cite{alur:tcs,henzinger95algorithmic,DBLP:conf/vmcai/PrabhakarDMV13}). 
This notion of abstraction is related to the notions of bisimilarity, equivalence, and implementation
used elsewhere in the literature. Roughly, if $A$ and $B$ are abstractions of each other then sometimes they are said to be observationally equivalent.

For models with continuous dynamics, it is makes sense to relax the notion of abstraction using a metric on the observable behaviors~\cite{GJP:IFAC2006,approxbisim}.  $B$ is said to be a $c$-abstraction of $A$, for some positive constant $c$, 
if every observable behavior of $A$ is within $c$ distance of some observable behavior of $B$,
where the distance is measured by some metric on the observables. 
In this paper we also look at time-bounded versions of abstraction.
$B$ is a $c$-abstraction of $A$ up to time $T$, 
if every observable behavior of $A$ of duration $T$ is within $c$ distance of 
some observable behavior of $B$ (also of duration $T$).

We can state properties about cyberphysical systems under attack with this relaxed notion.
Let $B$ be the nominal model  (without any attack) and $A_1$ be a model of the system under attack $1$. 
If we can prove that $B$ is a $c$-abstraction of $A_1$ then it follows that 
none of the observable behaviors deviate more that $c$ under attack $1$.
This gives a systematic way of classifying attacks with respect to their impact on deviation from ideal behavior.
If $B$ is {\em not} a $c$-abstraction of $A_2$---the model of the system under attack $2$---then it follows that 
attack $2$ is worse than attack $1$ in the sense that it causes a larger deviation from the ideal. 
A $c$-abstraction relation can also be used for reasoning about attack detectability and distinguishability.
If $B$ is a $c$-abstraction of  $A_1$ but our attack sensing mechanisms can only detect 
deviations in observable behavior from $B$ that are greater than $c$, 
then attack $1$ will go undetected.
If $B$ is also a $c$-abstraction of  $A_2$, then the same detection mechanism will also fail to 
distinguish the  two attacks.

The above discussion illustrates that many questions related to security and attacks can be formulated in terms of 
whether or not a model $B$ is an (relaxed) abstraction of another model $A$.
A building-block for our analytical framework is a semi-decision procedure for answering precisely this type of queries.
The procedure is sound, that is, whenever it terminates with an answer ($c$-abstraction or not)
the answer is correct. 
It is a semi-decision procedure because it is guaranteed to terminate, whenever the pair of models satisfy or violate the query robustly. 
Specifically, our contributions  are:
\begin{enumerate}[(a)]
\item
We formalize this quantitative notion of abstraction for models of dynamical systems
as the maximum distance from any trace of the concrete model $A_1$ to some trace of the abstract model $A_2$. 
\item For nonlinear ODE models, we present a semi-decision procedure for deciding 
if $A_2$ is a $c$-abstraction of $A_2$ up to a time bound, for any positive constant $c$. 
We show that the procedure is sound and it is guaranteed to terminate if either $A_2$ is at least a $\frac{c}{2}$-abstraction of $A_1$
or if there exists a trace of $A_1$ that is more than $2c$ distance away from all traces of $A_2$.
\item This semi-decision procedure and some of our earlier works for reachability~\cite{DMV:EMSOFT2013} 
use representations of reach sets of models. One of the contributions of this paper is the formalization of 
a natural data structure called {\em pipes\/} to represent simulation traces and reachable sets 
and identifying some of its key properties.
\item 
We present a procedure for computing over-approximations of 
unbounded time reach sets of individual models.
\end{enumerate}  

Checking equivalence of  two finite state machines---arguably the simplest class of models---is well-known to be  decidable.
The problem was shown to be decidable for deterministic push-down automata in the celebrated paper~\cite{senizergues1997equivalence}.
The same problem  becomes undecidable for finite state transducers~\cite{griffiths1968unsolvability} and 
nondeterministic pushdown automata.
For infinite state models that naturally capture computation and physics,
such as timed and hybrid automata~\cite{alur95algorithmic,henzinger95algorithmic,TIOAmon},
not only is equivalence checking undecidable, but so is the conceptually simpler 
problem of deciding if a single state can be reached by a given automaton. 
For models described by nonlinear ODEs, exactly computing the state reached from a single initial state at a given time is itself a hard problem. 
A sequence of recent results~\cite{donze2010breach,annpureddy2011s,DMV:EMSOFT2013,HuangM:hscc2012,HuangM:HSCC2014,DuggiralaJZM12}  circumvent these negative results by taking a more practical view of the reachability problem.
Specifically, they aim to compute over-approximations of the reach set over a bounded-time horizon.
Although some of these procedures require additional annotations of the models and provide  weaker soundness and completeness guarantees,
they point towards a practical way forward in automatically analyzing  reachability properties of moderately complex cyberphysical systems. 
The key characteristic of these approaches is that they combine static  model information (e.g. the differential equations and the text of the program, and not solutions or program runs)  
with dynamic information (e.g., possibly inaccurate numerical simulations or data from runtime logs), 
to compute precise over-approximations of bounded time reach sets. 
This static-dynamic analysis approach takes advantage of both static analysis techniques like propagating reach sets with 
dynamically generates information. 

This paper takes this static-dynamic analysis approach to checking abstraction relations. 
If an over-approximation of the reach set of $A$ is close the an over-approximation of the reach set of $B$ this 
means that  every behavior $\nu$ of $A$ is close to some over-approximation of $B$, but not this {\em does not\/} 
imply that $\vu$ is close to some actual behavior of $B$. 
Our  procedure (Algorithm~1) therefore has to take into account the precision of the 
over-approximations of $A$ and $B$ in deciding that each behavior of $A$ is indeed close to some (or far from all) behavior(s) of $B$. 
We also present a fixpoint procedure (Algorithm~2) which uses this static-dynamic approach to compute unbounded time reach sets.
For the sake of simplifying presentations, in this paper we presented the results for models of nonlinear dynamical systems, but these results can be extended to switched systems~\cite{DL2003} in a more or less straightforward fashion. 
Switched systems can capture commonly used time-triggered control systems which cover a large fraction of practical cyberphysical systems. 
Analogous extensions for reachability algorithms have been presented in~\cite{DMV:EMSOFT2013}.
The extension to hybrid models which can capture event-triggered interaction of software and continuous dynamics  
will be presented in a future paper.

\subsection{The Science}
\label{sec:science}

\begin{quote}
{\em ``I have observed stars of which the light, it can be proved, must take two million years to reach the earth.''}
---Sir William Herschel, British astronomer and telescope builder, 
having identified Uranus (1781), the first planet discovered since antiquity. 
\end{quote}

This paper presents a piece of  mathematical machinery (the semi-decision procedure) 
that is needed  for rigorous analysis of cyberphysical systems under different attacks. 
This procedure can be seen as a  {\em scientific instrument\/} that enables new types of attack impact measurements.
As we discussed above, abstraction is a central concept in any formal reasoning framework.
Abstraction relations in their quantitative form can be used to bound the distance from the set of observable 
behavior of one system to the set of  observable behaviors of another (ideal) system.
Thus, abstractions can give approximate measures of the deviation of an implementation from an idealized specification. 
Such measures can aid in the systematic evaluation of the effects of an attack  and 
in gaining understanding of different classes of attacks. 
In summary, the static-dynamic analysis techniques and specifically the semi-decision procedure 
presented in this paper can be seen as humble measuring instruments, 
but ones that could catalyze the science of security for CPS.

\section{Dynamical Systems}
\label{sec:model}

In this section, we present the modeling framework and some technical background used throughout the paper.
Some of the standard notations are left out for brevity. We refer the reader to the Appendix~\ref{app:basics} for details.

In this paper, we focus on models of dynamical systems with no inputs. 
Such models are also called autonomous or closed.  
An autonomous dynamical system is specified by a collection of ordinary differential equations (ODEs), an output mapping, and a set of initial states. 
\begin{definition}
\label{def:sys}
An $(n,m)$-dimensional {\em autonomous dynamical system\/} $A$ is a tuple $\langle  \Theta, f,g\rangle$ where
\begin{enumerate}[(i)]
\item $\Theta\subseteq \reals^n$ is a compact set of {\em initial states\/}; $\reals^n$ is the state space and it's elements are called {\em states\/}. 
\item $f:  \reals^n \mapsto \reals^n$ is a Lipschitz continuous function called the {\em dynamic mapping\/}. 
\item $g:  \reals^n \mapsto \reals^m$ is a Lipschitz continuous function called the {\em output mapping\/}. 
The output dimension of the system is $m$.
\label{it:f}
\end{enumerate}
\end{definition}
For a given initial state $\vx \in \Theta$, and a time duration $T \in \nnreals$, 
a {\em solution\/} (or trajectory) of $A$  is a pair of functions $(\xi_\vx,\nu_\vx)$: a state trajectory $\xi:[0,T] \mapsto \reals^n$
and an output trajectory $\nu:[0,T] \mapsto \reals^m$, such that 
(a) $\xi_\vx$ satisfies 
\begin{inparaenum}[(a)]
\item $\xi_\vx(0) = \vx$, 
\item for any $t\in [0,T]$, the time derivative of $\xi_\vx$ at $t$ satisfies the differential equation:
\begin{equation}
\label{eq:generaldynamics}
\dot \xi_\vx(t) = f(\xi_\vx(t)),
\end{equation}
\end{inparaenum}
And, (b) at each time instant $t\in [0,T]$, the output trajectory satisfies:
\begin{equation}
\label{eq:output}
\nu_\vx(t) = g(\xi_\vx(t)).
\end{equation}
Under the Lipschitz assumption (\ref{it:f}),  the differential equation~\eqref{eq:generaldynamics} admits a unique state trajectory defined by the initial state $\vx$
which in turn defines the output trajectory. 
When the initial state is clear from context, we will drop the suffix and write the trajectories as $\xi$ and $\nu$.
Given a state trajectory $\xi$ over $[0,T]$, the corresponding output trajectory or {\em trace\/} is defined in the obvious way as
$\nu(t) = g(\xi(t))$, for each $t \in [0, T]$. The same trace $\nu$, however, may come from a set of state trajectories. 
The set of all possible state trajectories and output trajectories of $A$ (from different initial states in $\Theta$) 
are denoted by  $\exec{A}$ and $\trace{A}$, respectively. A state $\vx \in \reals^n$ is said to be reachable if there 
exists $\vx' \in \Theta$ and $t \in \nnreals$ such that $\xi_{\vx'}(t) = \vx$. 
The set of all reachable states of $A$ is denoted by $\reach{A}$.
Variants of these notations are defined in the Appendix~\ref{app:basics}.

\paragraph*{Example}
We define a $(2,2)$-dimensional dynamical system. 
The set of initial states is defined by the rectangle $\Theta = [0.9,0.95] \times [1.5, 1.6]$.
The dynamic mapping is the nonlinear vector valued function:  
\[
f(x_1, x_2) = [1 + x^2y - 2.5x, -x^2y + 1.5x]. 
\]
And the output mapping is the vector valued identity function $g(x_1,x_2) = [x_1, x_2]$.
An over-approximation of the set of reachable states upto 10 time units 
(computed using the algorithm describes in~\cite{DMV:EMSOFT2013}) is shown in Figure~\ref{fig:bruss}.

\begin{figure}%
\includegraphics[width=\columnwidth]{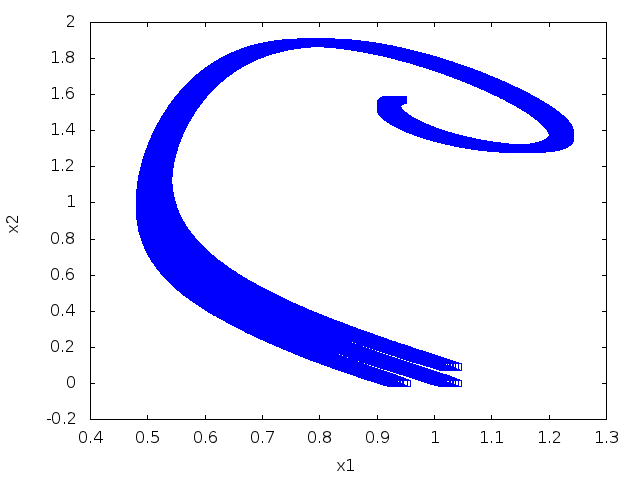}%
\caption{Reachable states and traces of the dynamical system in Example 1.}%
\label{fig:bruss}%
\end{figure}

\paragraph*{Trace metrics}
We define a metric on $d$ the set of traces of the same duration and dimension.
Given two traces $\nu_1, \nu_2$ of  duration $T$ and dimension $m$, 
we define 
\[
d(\nu_1, \nu_2) = \sup_{t \in [0,T]} |\nu_1(t) - \nu_2(t)|.
\] 
The distance from a set of traces $N_1$ to another set $N_2$ (with members of identical duration and dimension) 
is defined by the one-sided Hausdorff distance $d_H$ from $N_1$ to $N_2$.

\begin{definition}
\label{def:abs}
Given two autonomous dynamical systems $A_1$ and $A_2$ of identical output dimensions, 
a positive constant $c>0$ and a time bound $T>0$, $A_2$ is said to be $c$-{\em abstraction\/} of $A_1$ 
upto time $T$, if 
\[
d_H(\trace{A_1}(T), \trace{A_2}(T)) \leq c.
\] 
We write this as $A_1 \preceq_{c,T} A_2$.
\end{definition}
Thus, if $A_2$ is a $c$-abstraction of $A_1$, then for every output trace $\nu_1$ of $A_1$
there exists another output trace of $A_2$ which is differs from $\nu_1$ at each point in time by at most $c$.
Since, the definition only bounds  the one-sided Hausdorff distance, every trace of $A_2$ may not have a 
neighboring trace of $A_1$.
With $c= 0$, we recover the standard notion of abstraction, that is, $\trace{A_1} \subseteq \trace{A_2}$.
The next set of results follows immediately from the definitions and triangle inequality.

\begin{proposition}
\label{prop:c1c}
Let $A_1, A_2$ and $A_3$ be dynamical systems of identical output dimensions
and $c, c',  T$ be positive constants.
\begin{enumerate}
\item If $A_1 \preceq_{c,T} A_2$ then for any $c_1 \geq c$ and $T_1 \leq T$ $A_1 \preceq_{c_1,T_1} A_2$.
\item If $A_1 \preceq_{c,T} A_2$ and $A_2 \preceq_{c',T} A_3$ then $A_1 \preceq_{c+c',T} A_3$.
\end{enumerate}
\end{proposition}

\paragraph*{The decision problem}
The decision problem we solve in this paper takes as input 
a pair of autonomous dynamical systems $A_1$ and $A_2$ with identical output dimensions, annotations 
for these systems (namely, discrepancy functions which are to be defined in what follows),
a constant $c$ and a time bound $T$, and decides if $A_1 \preceq_{c,T} A_2$.
The computations performed by our algorithm uses {\em pipes\/} to represent sets of executions and traces.
In the next subsection, we define tubes and their properties.

\subsection{Working with Pipes}
\label{sec:tubes}

Pipes are used to represent sets of bounded traces and executions. 
Syntactically, an $n$-dimensional {\em pipe\/} is a sequence of {\em segments\/}
\[
\Pi = (P_0, t_0), (P_1, t_1), \ldots, (P_k, t_k),
\]
where in each segment $P_i$ is a subset of $\reals^n$ and $t_i \in \plreals$ and $t_i > t_{i-1}$.
The {\em duration\/} of the pipe is $\Pi.dur = t_k$ and its {\em length\/} is the number of segments $\Pi.len = k+1$.  

The semantics of a tube $\Pi$ is defined once we fix a dynamical system $A$.
It is the set of executions (or traces) of $A$ of duration $t_k$ defined as:
\begin{eqnarray*}
\sembrack{\Pi}_A = \{ \xi \in \exec{A}  & |&  \forall \ t \in [0,t_0], \xi(t) \in P_0,  \\
&& \forall \ 1 \leq i \leq \Pi.len,  t \in [t_{i-1},t_i], \xi(t) \in P_i\}.
\end{eqnarray*}
Our algorithms use tubes with finite representation---the sets $P_i$'s are compact
sets represented by polyhedra. 

We define $dia(\Pi) = \max_{i \in [\Pi.len]} dia(P_i)$ as the maximum diameter of any of the segments.
We say that two pipes $\Pi$ and $\Pi'$ are {\em comparable\/} if they have the same duration, length, and dimension 
and furthermore, for each $i \in [\Pi.len]$, $t_i = t_i'$. 
For two comparable pipes $\Pi, \Pi'$, we say that $\Pi$ is contained in $\Pi'$, denoted by $\Pi \subseteq \Pi'$, iff  
for each $i \in [\Pi.len]$, $P_i \subseteq P_i'$. 
The distance from $\Pi$ to $\Pi'$ is defined in the natural way by taking the maximum distance 
from the corresponding segments of $\Pi$ to those of $\Pi'$. 
\[
d_H(\Pi,\Pi') = \max_{i \in [\Pi.len]} d_H(P_i,P'_i).
\]
Obviously, $\Pi \subseteq \Pi'$ implies that $d_H(\Pi,\Pi') = 0$.

We say that two pipes are {\em disjoint\/}, denoted by $\Pi \cap \Pi' = \emptyset$, 
if and only if for each $i \in [\Pi.len]$, the corresponding segments are disjoint. 
That is, $P_i \cap P_i' = \emptyset$. 
%
The following straightforward propositions relate properties of pipes and the 
sets of executions (or traces) they represent.

\begin{proposition}
\label{prop:tube:contained}
Consider two comparable pipes $\Pi_1, \Pi_2$. 
If $\Pi_1 \subseteq \Pi_2$ then for any dynamical system $A$
$\sembrack{Pi_1}_A \subseteq \sembrack{Pi_2}_A$.
\end{proposition}

\begin{proposition}
\label{prop:tube:distance}
Consider two comparable pipes $\Pi_1, \Pi_2$. 
If $d_H(\Pi_1,\Pi_2) \geq c$ then for any two automata $A$ and $B$,
$d_H(\sembrack{Pi_1}_A, \sembrack{Pi_2}_B) \geq c$.
\end{proposition}


\subsection{Discrepancy Functions}
\label{sec:discrepancy}

Our decision procedure for $c$-abstractions will use numerical simulations (defined in Section~\ref{sec:algo}) and model annotations called {\em discrepancy functions\/}. Here we recall the definition of 
discrepancy functions which were introduced in~\cite{DMV:EMSOFT2013}.  
In that earlier paper we showed  that with discrepancy functions and numerical simulators 
we can obtain sound and relatively complete decision procedures for safety verification of nonlinear and switched dynamical system. 
Moreover, the software implementation of this approach proved to be scalable~\cite{C2E22013} .

Informally, a discrepancy function gives an upper bound on the distance between two trajectories as a function of the distance between their initial states and the time elapsed. 
\begin{definition}
\label{def:trajMetric}
A smooth function $V : \reals^{2n} \rightarrow \nnreals$ is called a {\em discrepancy function\/} for an $(n,m)$-dimensional dynamical system 
if and only if there are functions $\alpha_1,\alpha_2 \in \K_\infty$ and a uniformly continuous function $\beta: \reals^{2n} \times \reals \rightarrow \nnreals$ with $\beta(\vx_1,\vx_2,t) \rightarrow 0$ as
$|\vx_1-\vx_2| \rightarrow 0$ such that for any pair of states $\vx_1, \vx_2 \in \reals^n$:
\begin{align}
&\vx_1 \neq \vx_2 \iff V(\vx_1,\vx_2) > 0, \label{eq:nonzero} \\ 
&\alpha_1(|\vx_1 - \vx_2|) \leq V(\vx_1,\vx_2) \leq \alpha_2(|\vx_1 - \vx_2|) \mbox{and} \label{eq:alpha} \\
& \forall \ t > 0,\ V(\xi_{\vx_1}(t),\xi_{\vx_2}(t)) \leq \beta(\vx_1,\vx_2,t), \label{eq:beta}
\end{align}
\end{definition}
A tuple $(\alpha_1, \alpha_2, \beta)$ satisfying the above conditions is called a {\em witness \/} to the discrepancy function $V$. 
By discrepancy function of a dynamical system we will refer to $V$ as well as its witness interchangeably. 
Note that the output dimension $m$ has no bearing on the discrepancy function of the system.
The first condition requires that the function $V(\vx_1, \vx_2)$ vanishes to zero if and only if the first two arguments are identical. 
The second condition states that the value of $V(\vx_1, \vx_2)$ can be upper and lower-bounded by functions of the $\ell^2$ distance between $\vx_1$ and $\vx_2$. 
The final, and the more interesting, condition requires that the function $V$ applied to trajectories of $A$ at a time $t$ from a pair of initial states is  upper bounded and converges to $0$ as  $\vx_1$ converges to $\vx_2$.

For linear dynamical systems, discrepancy functions can be computed automatically by solving Lyapunov like equations, 
and in~\cite{DMV:EMSOFT2013} several strategies for proposed for nonlinear systems. 
Existing notions such as and Lipschitz constants, contraction metrics~\cite{ContractionNonlinear}, and incremental Lyapunov functions~\cite{lyapincremental,iiiss,GirardPT08} all yield discrepancy functions of varying quality. 

\section{A Semi-decision Procedure for Abstraction}
\label{sec:algo}

Our semi-decision procedure for $c$-abstractions will use numerical simulations of  the dynamical systems.
Given a closed dynamical system $A$ and a particular initial state $\vx \in \Theta$,  
for a step size $\tau>0$, validated ODE solvers (such as~\cite{capd,nedialkov1999,bouissou2006grklib}) compute a 
sequence of boxes (more generally polyhedra) $R_0,R_2,\dots,R_k \subseteq \reals^n$, 
such that for each $j\in[k]$, $t\in [(k-1)\tau,k\tau]$, $\xi_\vx(t)\in R_k$.
For a desired error bound $\epsilon>0$, by reducing the step size $\tau$, the diameter of $R_k$ can be made smaller than $\epsilon$. 
We define such simulations in terms of pipes below. 
\begin{definition}
\label{def:trace}
Given a dynamical system $A$, an initial state $\vx$, a time bound $T$, an error bound $\epsilon>0$, and time step $\tau>0$, 
a {\em $(\vx,T,\epsilon,\tau)$-simulation pipe} is a 
finite sequence $\phi = (R_0,t_0),(R_1,t_1),\dots,(R_k,t_k)$ where
\begin{enumerate}[(i)]
\item  $t_0=0$, $t_k=T$, and $\forall \ j\in [k]$, $t_{j}-t_{j-1} \leq  \tau$,
\item  $\forall j\in [k]$ and $\forall t\in [t_{j-1},t_j]$, $\xi_\vx(t)\in R_j$, and
\item $\forall j\in [k]$, $dia(R_j)\leq \epsilon$.
\end{enumerate}
\end{definition}
Our algorithm makes subroutine calls to a $\mathit{Simulate}$ function with these parameters 
which then returns a pipe with the above properties. 

The simulation pipe is then bloated using the discrepancy function of the dynamical system as follows.
\begin{definition}
\label{def:bloat}
Let $\mathit{sim} = (R_0,t_0),(R_1,t_1),\dots,(R_k,t_k)$ be a $(\vx,T,\epsilon,\tau)$-simulation pipe for a dynamical system $A$.
Suppose $V$ be a discrepancy function of $A$ with witness $(\alpha_1, \alpha_2, \beta)$. Then, for $\delta >0$, 
$\mathit{Bloat}(\mathit{sim}, \delta, V)$ is defined as the pipe 
$(P_0,t_0),\dots,(P_k,t_k)$ such that for each $j \in [k]$, 
\[
P_j = \{ \vx_1 \ | \exists \ \vx_2 \in R_j \ \wedge V(\vx_1, \vx_2) \leq e_j\},
\]
where 
\[
e_j = \sup_{t \in [t_{j-1}, t_j], \vx' \in B_\delta(\vx)} \beta(\vx, \vx', t).
\] 
\end{definition}
In other words, $e_j$ is an upper-bound on the value of $V$ for two executions $\xi_\vx$ and $\xi_{\vx'}$ 
starting from within $B_\delta(\vx)$ over the time interval $[t_{j-1}, t_j]$. 
And $P_j$ bloats $R_j$ to include all states $\vx_1$ for which there exists a state $\vx_2$ in $R_j$ 
with the discrepancy function bounded by $e_j$. 
Our algorithm makes subroutine calls to a $\mathit{Bloat}$ function which takes a simulation pipe, the function $\beta$ and the constant $\delta$ and returns the pipe $(P_0,t_0),\dots,(P_k,t_k)$  defined above. 

\begin{algorithm}[h!]
\caption{Deciding $c$-abstractions.}
\label{alg:main}
\SetKwInOut{Input}{input}
\SetKw{Goto}{go to}
\Input{$\A_1, V_1, \A_2, V_2, T, c$}
$\mathit{Init} \gets \Theta_1$\;
$\delta \gets \delta_0; \tau \gets \tau_0; \epsilon \gets \epsilon_0$\;
\While{$\mathit{Init} \neq \emptyset$}
{
	$X_1 \gets \mathit{Partition}(\mathit{Init}, \delta)$\;
	$X_2 \gets \mathit{Partition}(\Theta_2, \delta)$\;
	\For{$\vx_{10} \in X_1, \vx_{20} \in X_2$}
	{
		$\mathit{sim}[\vx_{10}] \gets \mathit{Simulate}(A_1,\vx_{10},\epsilon, T, \tau)$\;
		$\mathit{pipe}[\vx_{10}] \gets \mathit{Bloat}(\mathit{sim}[\vx_{10}],\delta,V_1)$\;				
		$\mathit{sim}[\vx_{20}] \gets \mathit{Simulate}(A_2,\vx_{20},\epsilon, T, \tau)$\;
		$\mathit{pipe}[\vx_{20}] \gets \mathit{Bloat}(\mathit{sim}[\vx_{20}],\delta,V_2)$\;
	}	
				\ForEach{$\vx_{10} \in X_1$}
				{
					\uIf{$\exists \ \vx_{20} \in X_2, d_H(\mathit{pipe}[\vx_{10}], \mathit{pipe}[\vx_{20}]) \leq \frac{c}{L_g}$ 
					$\wedge \mathit{dia}(\mathit{pipe}[\vx_{10}]) \leq \frac{c}{2L_g}$
					$\wedge \mathit{dia}(\mathit{pipe}[\vx_{20}]) \leq \frac{c}{2L_g}$ \label{ln:remove}}
					{
						$\mathit{Init} \gets \mathit{Init} \setminus B_\delta(\vx_{01})$\;
					}
					\uElseIf{$\forall \ \vx_{20} \in X_2, d_H(\mathit{pipe}[\vx_{10}], \mathit{pipe}[\vx_{20}]) \geq \frac{c}{S_g}$ 
					$\wedge \mathit{dia}(\mathit{pipe}[\vx_{10}]) \leq \frac{c}{2S_g}$ \label{ln:ce}}
					{
						\Return{\mbox{COUNTEREX} $\vx_{10}, \delta$} ;
					}
					\Else {$\delta \gets \frac{\delta}{2}; \tau \gets \frac{\tau}{2}; \epsilon \gets \frac{\epsilon}{2}$\;  \label{ln:refine}}
				}			
	}

\Return{$c$-\mbox{ABSTRACTION}} 
\end{algorithm}

\subsection{Description of the Algorithm}
Inside the while loop, first, two $\delta$-covers are computed for 
$\mathit{Init}$---a subset of the initial states $\Theta_1$ of $A_1$,  
and the set of initial states $\Theta_2$ of $A_2$
Next, in the first for loop, 
for each of the states $\vx_{10} \in X_1$ and $\vx_{20} \in X_2$  
in the respective covers, a $(\vx_{i0}, T, \epsilon, \tau)$-simulation pipe $\mathit{sim}[\vx_{i0}]$ 
is computed. Then this pipe is bloated with the parameter $\delta$ and the corresponding 
discrepancy function $V_i$.
The following proposition summarizes the main property of the bloated pipes.

\begin{proposition}
\label{prop:pipecontainsexecs}
For the dynamical system $\A_i, i \in \{1,2\}$ and constants $\delta, \epsilon, \tau$ and $T$,
$\exec{A_i}(B_\delta(\vx_{i0}),T) \subseteq \sembrack{\mathit{pipe}[\vx_{i0}]}$.
\end{proposition}
\begin{proof}
Let $\mathit{sim}[\vx_{i0}] = (R_0, t_0), \ldots, (R_k, t_k)$ and
 $\mathit{pipe}[\vx_{i0}] = (P_0, t_0), \ldots, (P_k, t_k)$.
We fix an initial state $\vx' \in B_\delta(\vx_{i0})$, and show that 
for any $t \leq t_k$, the state $\xi_\vx'(t)$ is contained in the set $P_j$,
where $t_{j-1} \leq t \leq t_{j}$. Let us fix $t$, which also fixes $t_{j-1}$ and $t_j$. 
From the definition of the {\em Simulation\/} (Definition~\ref{def:trace}) function we know that 
$\xi_{\vx_0}(t) \subseteq R_j$. 
And from Definition~\ref{def:bloat}, we know that since $\vx' \in B_\delta(\vx_{i0})$, 
the $V(\xi_{\vx'}(t), \xi_{\vx_{i0}}(t)) \leq \beta(\vx',\vx_{i0},t)$ and therefore  
$\xi_{\vx'}(t) \in P_j$. 
\end{proof}

\begin{corollary}
\label{cor:post}
For the dynamical system $\A_i, i \in \{1,2\}$ and constants $\delta, \epsilon, \tau$ and $T$,
\[
\reach{A_i}(\Theta_i,T) \subseteq \bigcup_{\vx_{i0} \in X_i} \bigcup_{j \in [T/\tau]} \mathit{pipe}[\vx_{i0}].P_j,
\]
here we use $\mathit{pipe}[\vx_{i0}].P_j$ to denote the subset of $\reals^n$ in the $j^{th}$ segment of the pipe 
$\mathit{pipe}[\vx_{i0}]$.
\end{corollary}

Every time a new set of bloated simulation pipes are computed, 
$\mathit{pipe}[\vx_{10}]$ for each $\vx_{10} \in \mathit{Init}$ 
and 
$\mathit{pipe}[\vx_{20}]$ for each $\vx_{10} \in X_2$,
the algorithm performs the following checks.
If there exists a $\mathit{pipe}[\vx_{10}]$ and a $\mathit{pipe}[\vx_{20}]$,
both  less than $c/2L_g$ in diameter and within $c/L_g$ distance then 
$B_{\delta}(\vx_{10})$ is eliminated from $\mathit{init}$.
Here $L_g$ is the Lipschitz constant and $S_g$ is the sensitivity constant of 
the common output function $g$.
If there exists a $\mathit{pipe}[\vx_{10}]$ such that for all the $\mathit{pipe}[\vx_{20}]$'s, $\vx_{20} \in X_2$, 
the diameter of the first is less than $c/2S_g$ and they are at least  $c/S_g$ distance 
away from each other, then $\vx_{10}$ (and $\delta$) is produced as a counter-example
to the $c$-abstraction. 
The while loop ends when $\mathit{Init}$ becomes empty.

\subsection{Soundness and Termination of Algorithm}
\label{subsec:analysis}
In this section, we prove the correctness of the algorithm.
We assume that the output mappings (the function $g$) is the same for the 
two models. 
\begin{theorem}
\label{thm:main1}
For automata with identical observation mappings, the algorithm is sound. \\
That is, if the output is $c$-$\mathit{ABSTRACTION}$, then, $A_2$ is a $c$-abstraction of $A_1$ upto time $T$,
and if the output is  $(\mathit{COUNTEREX}, \vx_{10}, \delta)$ then $A_2$ is not a $c$-abstraction of $A_1$.
In the latter case, all the traces of $A_1$ corresponding to executions starting from $B_\delta(\vx_0)$
are at least $c$ distance away from any trace of $A_2$.
\end{theorem}
\begin{proof}
For the first part, assume that the algorithm returns $c$-$\mathit{ABSTRACTION}$ and we 
will show that for any initial state $\vx \in \Theta_1$, there exists an initial state $\vx' \in \Theta_2$
such that $d(\nu_\vx,\nu'_{\vx'}) \leq c$. Here $\nu_\vx$ is the output trace of $A_1$ from $\vx$ and 
$\nu'$ is the output trace of $A_2$ from $\vx'$.

The algorithm returns $c$-$\mathit{ABSTRACTION}$ only when $\mathit{Init}$ becomes empty. 
This occurs when each initial state $\vx \in \Theta_1$ of $A_1$ is in the 
$\delta$-ball of some state $\vx_{10} \in \Theta_1$
such that $\vx_{10}$ is in a cover $X_1$ and satisfies the condition in Line~\ref{ln:remove}. 
It suffices to show that this condition $d_H(\mathit{pipe}[\vx_{10}],\mathit{pipe}[\vx_{20}]) \leq c/L_g$ 
implies that there exists $\vx' \in \Theta_2$, $d(\nu_\vx,\nu'_{\vx'}) \leq c$.
%

From Proposition~\ref{prop:pipecontainsexecs} it follows that 
for any $\vx' \in B_\delta(\vx_{20})$, and for any $t \in [0,T]$, 
$|\xi_\vx(t) - \xi'_{\vx'}(t)| \leq c/L_g$.
Let us fix a $\vx' \in B_\delta(\vx_{20})$. 
Then, the distance between the corresponding traces is:
\begin{eqnarray*}
d(\nu_\vx, \nu'_{\vx'}) &=& \sup_{t \in [0,T]} |\nu_\vx(t) - \nu'_{\vx'}(t)| \\
&=& \sup_{t \in [0,T]} |g(\xi_\vx(t)) - g(\xi'_{\vx'}(t))| \\
&\leq&  L_g \sup_{t \in [0,T]} |\xi_\vx(t) - \xi'_{\vx'}(t)| \\
&\leq&  L_g \frac{c}{L_g} = c.
\end{eqnarray*}
Here we have used the assumption that the two systems have the same output mapping $g$ 
and recall that $L_g$ is the Lipschitz constant of this mapping.

For the second part, we assume that the algorithm returns $\mathit{COUNTEREX}$
and show that there exists a trace $\nu_{\vx_{10}}$ of $A_1$ which is at least 
$c$ distance away from all traces of $A_2$.
Examining the algorithm, output of  $\mathit{COUNTEREX}$  occurs if there exists a 
constant $\delta > 0$ and $\vx_{10} \in \Theta_1$, 
such that for any initial state $\vx_{20}$ of $A_2$, 
\[
d_H(\mathit{pipe}[\vx_{10}], \mathit{pipe}[\vx_{20}]) \geq \frac{c}{S_g},
\]
where $S_g$ is the sensitivity of the output mapping.
From Proposition~\ref{prop:tube:distance} it follows that 
for any $\vx' \in \Theta_2$, for all $t \in [0, T]$, 
$\xi_{\vx_{10}}(t) - \xi'_{\vx'}(t) \geq \frac{c}{S_g}$.
Now, consider the distance between the corresponding traces:
\begin{eqnarray*}
d(\nu_{\vx_{10}}, \nu'_{\vx'}) &=& \sup_{t \in [0,T]} |g(\xi_{\vx_{10}}(t)) - g(\xi'_{\vx'}(t))| \\
&\geq& \sup_{t \in [0,T]}  S_g |\xi_{\vx_{0}}(t) - \xi'_{\vx'}(t)| \\
&\geq&  S_g \frac{c}{S_g} = c.
\end{eqnarray*}
Thus, the traces are at least $c$ apart, and $A_2$ is not a $c$-abstraction of $A_1$.
Here again we have used the assumption that the two systems have the same output mapping $g$ 
 $S_g$ is the Lipschitz constant of this mapping.

\end{proof}

Next, we prove that the algorithm terminates provided 
either (a) $A_1 \preceq_{\frac{c}{2},T} A_2$ 
or (b) $d_H(\trace{A_1},\trace{A_2}) > 2c$.

\begin{theorem}
\label{thm:main2}
The $c$-abstraction algorithm terminates  
either if $A_2$ is a $c_1$-abstraction of $A_1$ for any $c_1 < \frac{c}{2}$
or if there exists a trace of $A_1$ which is $c_2$ distance away from all traces of $A_2$, for some $c_2 > 2c$.
\end{theorem}

\begin{proof}
Assume without loss of generality that $L_g <  2S_g$.
For the first part, assume that the $A_2$ is a $\frac{c}{2}$-abstraction of $A_1$.
First, we will show that Line~\ref{ln:ce} returning a counter-example is never executed.
For the sake of contradiction, let us assume that this line is executed. 
Then $\mathit{dia}(\mathit{pipe}[\vx]) \leq \frac{c}{2S_g}.$ 
Also, for any execution $\xi_\vx$ of $A_1$ there exists an execution $\xi'_{\vx'}$ of $A_2$ 
such that for any $t \in [0,T]$, $|\xi_{\vx}(t) - \xi'_{\vx'}(t)| < \frac{c}{2S_g}$
(this is because $A_1 \preceq_{\frac{c}{2},T} A_2$). 
Thus, from any pipe containing $\xi_\vx$, the distance to any pipe containing any $\xi'_{\vx'}$ 
is less than $\frac{c}{S_g}$. This violates the precondition for returning a counter-example.

It suffices to show that every initial state $\vx \in \Theta_1$ is eventually removed from $\mathit{Init}$
in Line~\ref{ln:remove}. Let us fix an execution $\xi'_{\vx'}$ of $A_2$ 
such that for any $t \in [0,T]$, $|\xi_{\vx}(t) - \xi'_{\vx'}(t)| < \frac{c}{2S_g}$. 
Under the above conditions, in each iteration of the while loop, 
$\delta, \tau$, and $\epsilon$ are halved in Line~\ref{ln:refine}.
From the Definition~\ref{def:trace} the diameter $\mathit{dia}(\mathit{sim}[\vx_{10}]) \leq \epsilon$
for $\vx_{10}$ with $|\vx_{10} - \vx| \leq \delta$.
Similarly,  the diameter $\mathit{dia}(\mathit{sim}[\vx_{20}]) \leq \epsilon$
for $\vx_{20}$ with $|\vx_{20} - \vx'| \leq \delta$.
Notice that as these parameters decrease, from the definition of the discrepancy function, 
$\beta_i(\vx,\vx',.) \rightarrow 0$. Thus, the distance between the bloated tubes containing  $\xi_{\vx}$ and 
$\xi'_{\vx'}$ also converge to $\frac{c}{2S_g}$.
With the assumption that $L_g < 2S_g$, it follows that eventually two points $\vx_{10} \in X_1$
and $\vx_{20} \in X_2$ will satisfy the condition ($d_H(\mathit{pipe}[\vx_{10}], \mathit{pipe}[\vx_{20}]) \leq \frac{c}{L_g}$) in Line~\ref{ln:remove} with $\vx' \in B_\delta(\vx_{20})$, and more importantly, $\vx \in B_{\delta}(\vx_{10})$.

For the second part, suppose there exists a trace $\nu_\vx$ of $A_1$ such that 
for any trace $\nu'_{\vx'}$ of $A_2$, $d(\nu_\vx, \nu_\vx) > 2c$. 
Then we know that for each $t \in [0,T]$, $|\xi_\vx(t) - \xi'_{\vx'}(t)| > 2c/L_g$.
For the sake of contradiction, let us assume that $\vx$ is eliminated from $\mathit{Init}$ 
in Line~\ref{ln:remove}. Then there must exist $\mathit{pipe}[\vx_{10}]$ containing $\xi_\vx$ 
  $\mathit{pipe}[\vx_{20}]$ containing $\xi'_{\vx'}$ with $\mathit{dia}(\mathit{pipe}[\vx_{10}]) \leq c/2L_g$
	and $\mathit{dia}(\mathit{pipe}[\vx_{20}]) \leq c/2L_g$.
	Then, $d_H(\mathit{pipe}[\vx_{10}],\mathit{pipe}[\vx_{10}]) >  2c/L_g - c/2L_g -  c/2L_g = c/L_g$
	which contradicts the first condition in Line~\ref{ln:remove}.
	
So, $\vx$ is never eliminated from $\mathit{Init}$. 
Analogous to the argument presented for the first part, the pipe computed containing $\xi_\vx$ 
become smaller and smaller in diameter as $\tau, \epsilon$ and $\delta$ are reduced
and the  pipes computed containing all the executions of $A_2$ starting from $\Theta_2$, including $\xi'_{\vx'}$ also 
become smaller. Eventually,  $d_H(\mathit{pipe}[\vx_{10}], \mathit{pipe}[\vx_{20}]) > 2c/L_g$
as for each $d(\xi_\vx, \xi'_{\vx'}| > 2c/L_g$. 
At this point the condition in Line~\ref{ln:ce} becomes true and $\vx$ is produced with the 
$\mathit{COUNTEEX}$ output.
\end{proof}
Thus, there is a range of values of the distance between the sets of traces 
$d_H(\trace{A_1}, \trace{A_2})$ in $[\frac{c}{2}, 2C]$ where the algorithm is not guaranteed to terminate. 
This range can be made arbitrarily small by choosing small values of $c$.

\section{Unbounded time Extension}
\label{sec:unbounded}

In the previous section, we presented a semi-decision procedure for reasoning about 
bounded-time abstraction relations between models of cyberphysical systems.
Since cyberphysical systems typically run for long time horizons, ideally 
we would like to perform unbounded time analysis. The following procedure
uses bound-time simulations and attempts to compute an over-approximation of the 
unbounded-time reach set of a dynamical system. 

The algorithm adapts a standard fixpoint procedure to now work with 
our simulation-based technique for computing reach set approximations.
The set $\mathit{newreach}$ stores the newly discovered reachable states
and $\mathit{reach}$ accumulates all the reachable states.
Both are initialized to $\Theta$.
The while loop iterates until no newly reachable states are discovered;
at that point $\mathit{reach}$ is produced as the output. 
Inside the while loop, $\mathit{newreach}$ is $\delta$-partitioned.
Then, as in  Algorithm~1, an array of $(\vx_1, k\tau, \epsilon, \tau)$-simulations 
$\mathit{sim}[\vx_1]$ are computed for each $\vx_1 \in X_1$
and they are bloated to compute the array of pipes $\mathit{pipe}[\vx_1]$.
The union of the segments in all these pipes give the set $\mathit{post}$
and the $\mathit{newreach}$ set is obtained by subtracting $\mathit{reach}$ from $\mathit{post}$.

\begin{algorithm}[h!]
\caption{Unbounded time reachability.}
\label{alg:main}
\SetKwInOut{Input}{input}
\SetKw{Goto}{go to}
\Input{$\A, V, k, \tau, \delta, \epsilon$}
$\mathit{newreach} \gets \Theta$\;
$\mathit{reach} \gets \Theta$\;
\While{$\mathit{newreach} \neq \emptyset$}
{
$X_1 \gets \mathit{Partition}(\mathit{newreach}, \delta)$\;
\For{$\vx_{1} \in X_1$}
	{
		$\mathit{sim}[\vx_{1}] \gets \mathit{Simulate}(A,\vx_{1},\epsilon, k\tau, \tau)$\;
		$\mathit{pipe}[\vx_{1}] \gets \mathit{Bloat}(\mathit{sim}[\vx_{1}],\delta,V)$\;				
	}	
	$\mathit{post} \gets  \cup_{i \in[k]}\cup_{\vx_1 \in X_1} \mathit{pipe}[\vx_{1}].P_i$ \;  \label{ln:reach1}
	$\mathit{newreach} \gets \mathit{post} \setminus \mathit{reach}$\;
	$\mathit{reach} \gets \mathit{reach} \cup \mathit{newreach}$\;
}
\Return{$\mathit{reach}$}
\end{algorithm}

The following theorem states that if the above algorithm returns a set of states $R$
then this set is an over-approximation of the {\em unbounded\/} reach set of the dynamical system $A$.
\begin{theorem}
\label{thm:unbounded} 
If Algorithm~2 returns a set of states $R$ then $\reach{A}(\Theta) \subseteq R$.
\end{theorem}
\paragraph*{Proof sketch}
From Corollary~\ref{cor:post} it follows that in each iteration of the while loop 
$\reach{A}(\mathit{newreach}, k\tau) \subseteq \mathit{post}$, that is the set computed 
using simulations and bloating in Line~\ref{ln:reach1}.
The set $\mathit{newreach}$ is updated to be an over-approximation of the states that are reached 
for the first time in the current iteration.
A simple induction on the number of iterations show that at the $i^{th}$ iteration,
 $\mathit{reach}$ contains all states that are reachable from $\Theta$ in $i(k\tau)$ time.
The computation halts in an iteration when no new reachable states are discovered
and the corresponding output $\mathit{reach}$ is the least fixpoint of the algorithm 
containing $\Theta$ and therefore it  over-approximates the unbounded-time reach set from $\Theta$.

\section{Discussions}
\label{sec:dis}

The simulation-based reachability algorithms~\cite{DMV:EMSOFT2013,donze2010breach,staliro-tool-paper,zhthesis,HuangM:hscc2012} 
have provided a general and scalable building-block for analysis of nonlinear, switched, and hybrid models. Since simulation-based analysis can be made embarrassingly parallel, these approaches can scale to real-world models with dozens and possibly hundreds of continuous dimensions. 
This paper takes this static-dynamic analysis approach to checking abstraction relations. 
As we discussed in the introduction, computing reach set over-approximations are not sufficient for 
reasoning about abstraction relations. 
Our  procedure takes into account the precision of the 
over-approximations  in deciding that each behavior of $A$ is indeed close to some behavior 
of $B$ or that there exits a behavior of $A$ that is far from all behaviors of $B$. 
For the sake of simplifying presentations, in this paper we presented the results for models of nonlinear dynamical systems, but these results can be extended to switched systems~\cite{DL2003} in a more or less straightforward fashion (see~\cite{DMV:EMSOFT2013}
for analogous extensions for reachability algorithms).
This work suggests several directions for future research in developing 
new notions of abstraction, corresponding decision procedures, and in extending them to be applicable to broader classes of 
models that arise in analysis of cyberphysical systems under attacks.

\subsection {Future Research Directions}
\label{sec:future}

\paragraph*{Switched system models and models with inputs}
The switched system~\cite{DL2003} formalism is useful where it suffices to view the software or the adversary as something that only changes the continuous dynamics. They are useful for modeling  time-triggered control systems and timing-based attacks.
A switched system is described by a collection of dynamical systems (Definition~\ref{def:sys}) and a piece-wise constant switching signal that determines which particular ODE from the collection that governs the evolution of the system at a given time. 
A timing attack can be modeled as altered switching signal (as well as the changed dynamics).
One nice property of switched system models is that the executions $\xi$ are continuous functions of time just like ODEs.
If all the ODEs are equipped with discrepancy functions, then we show in~\cite{DMV:EMSOFT2013} that it is possible to compute reach set over-approximations for a set of switching signals by partitioning both the initial set and the set of switching signals. 
This technique essentially works also for analyzing abstraction relation between switched models. 

Switched and ODE models with inputs will enable us to model open cyberphysical systems and adversaries that feed bad inputs to such systems. The main challenge here is reasoning about the distance between trajectories that start from different initial states, as well as, 
experience different input signals. In our recent paper~\cite{HuangM:HSCC2014} we have used an input-to-state discrepancy function to reason about reachability of such models and a similar approach can work for abstractions. 

\paragraph*{Nondeterministic models}
All of the above models are deterministic once the initial states, the switching signal, and the input signal are specified. 
In order to apply out analytical framework to a broader class of system models and attacks, we have to develop decision procedures for hybrid model with nondeterministic transitions as well as nondeterministic dynamics. For the latter case, the results for reach set over-approximation presented in~\cite{zhthesis} could provide a starting point.


\appendix

\section{Basic Definitions and Notations}
\label{app:basics}

For a natural number $n\in \naturals$, $[n]$ is the set $\{1,2,\dots, n\}$.
For a sequence $A$ of objects of any type with $n$ elements, we refer to the $i^{th}$ element, $i \leq n$ by $A_i$.
For a real-valued vector $\vx$, $|\vx|$ denotes the  $\ell^2$-norm unless otherwise specified.
The diameter of a  compact set $R \subseteq \reals^n$, $dia(R)$ is defined as the maximum distance between any two 
points in it: $dia(R) = sup_{\vx, \vx' \in R} |\vx - \vx'|$.

\paragraph*{Variable valuations} 
Let $V$ be a finite set of real-valued variables. Variables are names for state and input components. 
A {\em valuation $\vv$ for $V$} is a function mapping each variable name to its value in $\reals$. The set of valuations for $V$ is denoted by $\val{V}$.
Valuations can be viewed as vectors in $\reals^{|V|}$ dimensional space with by fixing some arbitrary ordering on variables.
$B_\delta(\vv) \subseteq \val{V}$ is the closed ball of valuations with radius $\delta$ centered at $\vv$.
%
The notions of continuity, differentiability, and integration are lifted to functions defined over sets of valuations in the usual way. 

For any function $f:A\mapsto B$ and a set $S \subseteq A$, $f \restr S$ is the restriction of $f$ to $S$.
That is, $(f \restr S)(s = f(s)$ for each $s \in S$.
So, for a variable $v \in V$ and a valuation $\vv \in \val{V}$, 
$\vv \restr v$ is the function mapping $\{v\}$ to the value $\vv(v)$. 
A function $f:A \mapsto \reals$ is {\em Lipschitz} if there exists a constant $L\geq 0$---called the {\em Lipschitz constant}---such that for all $a_1, a_2\in A$  $|f(a_1)-f(a_2)| \leq L|a_1-a_2|$. 
We define a function $f$ to have {\em sensitivity\/} of $S_f$ if 
for all $a_1, a_2\in A$  $|f(a_1)-f(a_2)| \geq S_f|a_1-a_2|$. 
A continuous function $\alpha: \nnreals \mapsto \nnreals$ is in the {\em class of $\K$ functions} if $\alpha(0) = 0$  and it is strictly increasing. 
 Class $\K$ functions are closed under composition and inversion. 
A class $\K$ function $\alpha$  is a {\em class $\K_\infty$} function if $\alpha(x) \rightarrow \infty$ as $x\rightarrow \infty$.
%
A continuous function $\beta: \nnreals\times \nnreals \mapsto \nnreals$ is called a {\em class $\K\L$ function} if for any $t$, 
$\beta(x,t)$ is a class $\K$ function in $x$ and for any $x$, $\beta(x,t)\rightarrow 0$ as $t\rightarrow \infty$.

\paragraph*{Trajectories} 
Trajectories model the continuous evolution of variable valuations over time. 
A {\em trajectory\/} for $V$ is a differentiable function $\tau: \nnreals \mapsto \val{V}$. 
The set of all possible trajectories for $V$ is denoted by $\traj{V}$.
For any function $f:C \mapsto [A\mapsto B]$ and a set $S \subseteq A$, $f \restrrange S$ is the restriction of $f(c)$ to $S$.
That is, $(f \restrrange S)(c) = f(c) \restr S$   for each $c \in C$.
In particular, for a variable $v \in V$ and a trajectory $\tau \in \traj{V}$, 
$\tau \restrrange v$ is the trajectory of $v$ defined by $\tau$.

\paragraph*{Dynamical systems}
The set of all trajectories of $A$ with respect to a set of initial states $\Theta' \subseteq \val{X}$ and a set of  is denoted by $\traj{A, \Theta'}$. 
The components of dynamical system $A$ and $A_i$ are denoted by $X_A, \Theta_A, f_\A$ and 
$X_i, \Theta_i, f_i$, respectively. We will drop the subscripts when they are clear from context.
The set of all possible state trajectories and output trajectories of $A$ (from different initial states in $\Theta$) 
are denoted by  $\exec{A}$ and $\trace{A}$, respectively.
The set of executions (and traces) from the set of initial states $\Theta$ and upto time bound $T$ 
is denoted by  $\exec{A}(\Theta, T)$ (and  $\trace{A}(\Theta, T)$, repectively).
A state $\vx \in \reals^n$ is {\em reachable\/} if there exists and execution $\xi$ and a time $t$
such that $\xi(t) = \vx$. 
The set of reachable states from initial set $\Theta$ within time $T$ is denoted by $\reach{A}(\Theta, T)$.

\bibliographystyle{abbrv}
\bibliography{iss,sayan1}

\end{document}